\documentclass[pdflatex,sn-mathphys]{sn-jnl}

\jyear{2021}%

\theoremstyle{thmstyleone}%
\newtheorem{theorem}{Theorem}
\newtheorem{proposition}{Proposition}
\newtheorem{corollary}{Corollary}

\usepackage{mathtools}
\usepackage{algorithm}
\usepackage{amsmath}
\usepackage{amssymb}
\usepackage{tabularx}
\usepackage{float}
\usepackage{lscape}
\usepackage{multirow}
\usepackage{algpseudocode}

\newtheorem{lemma}{Lemma}

\theoremstyle{thmstyletwo}%
\newtheorem{example}{Example}%

\theoremstyle{thmstylethree}%

\begin{document}

\title[Article Title]{Quaternary Conjucyclic Codes with an Application to EAQEC Codes}


\author*[1]{\fnm{Md Ajaharul} \sur{Hossain}}\email{mdajaharul@iiitnr.edu.in}

\author[1]{\fnm{Ramakrishna} \sur{Bandi}}\email{ramakrishna@iiitnr.edu.in}
\equalcont{These authors contributed equally to this work.}


\affil*[1]{\orgdiv{Department of Science and Applied Mathematics}, \orgname{IIIT Naya Raipur}, \orgaddress{ \city{Atal Nagar Nava Raipur}, \postcode{493661}, \state{Chhattisgarh}, \country{India}}}




\abstract{

Conjucyclic codes are part of a family of codes that includes cyclic, constacyclic, and quasi-cyclic codes, among others. Despite their importance in quantum error correction, they have not received much attention in the literature. This paper focuses on additive conjucyclic (ACC) codes over $\mathbb{F}_4$ and investigates their properties. Specifically, we derive the duals of ACC codes using a trace inner product and obtain the trace hull and its dimension. Also, establish a necessary and sufficient condition for an additive code to have a complementary dual (ACD). Additionally, we identify a necessary condition for an additive conjucyclic complementary pair of codes over $\mathbb{F}_4$. Furthermore, we show that the trace code of an ACC code is cyclic and provide a condition for the trace code of an ACC code to be LCD. To demonstrate the practical application of our findings, we construct some good entanglement-assisted quantum error-correcting (EAQEC) codes using the trace code of ACC codes.
}




\maketitle

\section{Introduction}\label{sec1}

There are many scenarios in which time is money, and two of the most prominent examples are quantum computing and quantum communication. Classical computers and even supercomputers cannot solve many computational problems in a reasonable amount of time, but quantum computers can do so quickly. This will have a significant impact on our technological world. However, quantum computing and communication require delicate handling of quantum states, which are sensitive to hardware imperfections and environmental interference. Decoherence is a natural phenomenon in quantum systems that can cause errors. Quantum error-correcting codes are crucial for preserving coherent states against such noise and interference. Over the past decade, the theory of quantum error-correcting codes has made significant progress. Several constructions, such as the CSS and geometric constructions, are available in the literature \cite{ST1996, C1997, KR2005, BI2008, BI2014}. The CSS construction leads to linear quantum codes, also known as stabilizer codes. Calderbank and Shor \cite{Cal96}, and Steane \cite{ST1996} independently proved the existence of good quantum error-correcting codes and developed the CSS construction (also known as dual-containing property) to construct quantum stabilizer codes from classical error-correcting codes. However, the dual-containing property can be a hindrance to constructing a large class of codes, such as LDPC codes, since codes with a larger minimum distance with the dual-containing property are rare. Additionally, the theory behind stabilizer codes makes it challenging to construct examples of stabilizer codes.

The most significant class of linear quantum codes is the Entanglement Assisted Quantum Error Correcting (EAQEC) Codes, which were introduced by Brun et al. in \cite{TD06}. These codes have a significant advantage over others as they cover a larger family of classical codes and can be used to generate quantum codes from them. EAQEC codes relax the self-orthogonality criterion for linear classical codes to construct quantum codes using pre-shared entangled bits (ebits). In \cite{Fan16}, Fan et al. provided a construction of entanglement-assisted quantum maximum distance separable (EAQ-MDS) codes with a small number of pre-shared maximally entangled states. Furthermore, Qian and Zhang \cite{Qian2015} constructed maximal-entanglement EAQEC codes and proved the existence of asymptotically good EAQEC codes in the binary case. In recent years, there has been significant interest among researchers in constructing EAQEC Codes \cite{ch17,ch18, PhysRevA.103.L020601, G2018, luo2019, WB08,hossain2023linear} from classical error-correcting codes.

In addition to the well-studied linear codes, Delsarte et al. presented an important family of codes known as additive codes in their work on association schemes \cite{DL73}. Generally, additive codes are subgroups of the underlying abelian group. In \cite{H07}, Huffman provided an algebraic structure for additive cyclic codes over $\mathbb{F}_4$. Various additive codes have been well-studied, and their applications in quantum error correction and secret sharing schemes have been explored in works such as \cite{CRS, KL17, BN20, HC13, EM11, KP03}. Calderbank et al. introduced a new family of codes called conjucyclic codes over quaternary fields in \cite{CRS} and established a direct link between quaternary additive codes and binary quantum error-correcting codes. A linear code $C$ of length $n$ over $\mathbb{F}_{q^2}$ is called a conjucyclic code if for any codeword $c=(c_0,c_1,\ldots, c_{n-1})$ in $C$, the conjucyclic shift $(\Bar{c}_{n-1}, c_0, \ldots, c_{n-2})$,  where $\Bar{c}_{n-1}$ is the conjugate of $c_{n-1}$ in $\mathbb{F}_{q^2}$, also belongs to $C$. Although conjucyclic codes are closed under conjugated cyclic shift and belong to the cyclic, quasi-cyclic, and constacyclic code families, they have a structural disadvantage compared to other codes. Since they lack a canonical polynomial representation, studying them has proven difficult and remains an open research problem. Recently, in \cite{ACD}, Abualrub et al. used a linear algebra approach to investigate the algebraic structure of conjucyclic codes over Quaternary fields. They established an isomorphic map between a conjucyclic code of length $n$ and a cyclic code of length $2n$, simplifying the analysis of conjucyclic codes and enabling them to be viewed through the lens of cyclic codes. Since the only conjucyclic codes that are linear are the trivial codes, additive conjucyclic (ACC) codes are being studied, particularly over $\mathbb{F}_4$, for obvious reasons, \cite{ACD}. Lv and Li also studied additive conjucyclic codes over the finite field $\mathbb{F}_{q^2}$ in \cite{LL}.

Although Conjucyclic codes have not been extensively explored, their potential applications in quantum error correction have sparked the interest of researchers \cite{CRS}. The authors of \cite{CRS} noted the lack of good ACC codes over $\mathbb{F}_4$, which motivated our search for good ACC codes and quantum codes through ACC codes. This paper focuses on conjucyclic codes in relation to the trace dual and establishes necessary and sufficient conditions for an ACC code to have ACD. We also construct linear codes from ACC codes, which exhibit excellent parameters. Finally, we present some good EAQEC codes obtained from conjucyclic codes.

\section{Preliminaries}
\label{Sec.2}
Throughout this paper, we use $\mathbb{F}_4$ to denote the quaternary field of order $4$. Specifically, we define $\mathbb{F}_4$ as $\{0,1,\omega,\omega^2=\omega+1\}$, where $\omega$ is a primitive root of the minimal polynomial $x^2+x+1=0$, and thus $\omega^3=1$. Two elements $\alpha$ and $\beta$ of a finite field are said to be conjugates if they are the roots of the same minimal polynomial. Therefore, the conjugate of $\omega \in \mathbb{F}_4$ is $\omega^2=\omega+1$, and vice versa. We denote the conjugate of $\alpha \in \mathbb{F}_4$ as $\Bar{\alpha}$.

A linear code of length $n$ over $\mathbb{F}_4$ is a subspace of $\mathbb{F}_4^n$, and its algebraic structure and properties are identical to those of a subspace. An $(n,2^k)$ additive code over $\mathbb{F}_4$ is an additive subgroup of $\mathbb{F}_4^n$ with $2^k$ elements. Unlike linear codes, additive codes do not support scalar multiplications. An additive code $C$ of length $n$ over $\mathbb{F}_4$ is an $\mathbb{F}_2$-subspace of $\mathbb{F}_4^n$. So we denote the ${\mathbb{F}_2}$-dimension of an additive code $C$ over $\mathbb{F}_4$ as $\operatorname{dim}_{\mathbb{F}_2}(C)$, which is the cardinality of the basis set for $C$ over $\mathbb{F}_2$. An $\mathbb{F}_2$-base of $C$ is called a generator matrix of $C$. Note that the $\mathbb{F}_2$ dimension of an additive code $C$ can be greater than the length of the code $C$, unlike a linear code. To illustrate, consider an additive code $C$ over $\mathbb{F}_4$ defined as $C=\{(0,0), (1,0),(0,1), (1,1), (\omega,0), (1+\omega,0), (\omega,1), (1+\omega,1)\}$. The generator matrix of $C$ is $G=\begin{bmatrix} 1 & 0 & \omega \\ 0 & 1 & 0 \end{bmatrix}^T$. The code $C$ has $\mathbb{F}_2$ dimension 3, i.e., $\operatorname{dim}_{\mathbb{F}_2}(C)=3$. Unless otherwise specified, all codes discussed over $\mathbb{F}_4$ in this paper are assumed to be additive codes.

We define a trace mapping $Tr: \mathbb{F}_4 \mapsto \mathbb{F}_2$ as $Tr(\alpha)=\alpha+\Bar{\alpha}$. This mapping is naturally extended to $\mathbb{F}_4^n$ as $Tr(a)=(Tr(a_0),Tr(a_1), \ldots, Tr(a_{n-1}))$, where $a=(a_0,a_1, \ldots, a_{n-1})$. In this paper, we consider two types of inner products: the Euclidean inner product and the trace inner product. The Euclidean inner product of two elements $a=(a_0,a_1, \ldots, a_{n-1})$ and $b=(b_0,b_1, \ldots, b_{n-1})$ in $\mathbb{F}_4^n$ is denoted by $[a,b]$ and is defined as $[a,b]=\sum_{i=0}^{n-1} a_ib_i$. However, the Euclidean inner product is not useful for studying additive codes since the dual code $C^\perp= \{x \in \mathbb{F}_4^n~:~ [x,c]=0 ~\forall c \in C \}$ of an additive code $C$ may not be an additive code. Therefore, a new inner product, the trace inner product, was introduced in \cite{ACD}.

The trace inner product of $a$ and $b$ in $\mathbb{F}_4$ is denoted by $\langle a, b \rangle$ and is defined as $\langle a, b\rangle= Tr([a, b])$. The trace dual of an additive code $C$ is denoted by $C^{\perp_{Tr}}$ and defined as follows: $C^{\perp_{Tr}}=\{x\in F_q^n:(\forall c\in C)( \langle x,c \rangle=Tr([x,~c])=0)\}$. We have $x\in {(C^{\perp_{Tr}})}^{\perp_{Tr}}$ if and only if $\langle x,c \rangle=Tr([x,~c])=0$ for all $c\in C^{\perp_{Tr}}$ if and only if $x\in C$. Therefore ${(C^{\perp_{Tr}})}^{\perp_{Tr}}=C$. The trace dual $C^{\perp_{Tr}}$ of $C$ is also an additive code, and the $\mathbb{F}_2$ basis of $C^{\perp_{Tr}}$ is called a parity check matrix of $C$. An additive code $C$ is said to be self-orthogonal with respect to its trace dual if $C\subseteq C^{\perp_{Tr}}$, and it is called self-dual if $C=C^{\perp_{Tr}}$. Finally, two additive codes $C_1$ and $C_2$ of length $n$ over $\mathbb{F}_4$ are called an additive complementary pair (ACP) if $C_1\cap C_2={0}$ and $C_1+C_2=\mathbb{F}_4^n$. Setting $C_1=C$ and $C_2=C^{\perp_{Tr}}$, we say that $C$ is an ACD code, that is, $C$ is an ACD code when $C$ and its trace dual $C^{\perp_{Tr}}$ meet trivially, i.e., $C\cap C^{\perp_{Tr}}={0}$.

Let $A = (a_{ij})_{m \times n}$ be a matrix over $\mathbb{F}_4$. Then we define the conjugate and trace of matrix $A$ as $\Bar{A} = (\Bar{a}_{ij})_{m \times n}$ and $Tr(A) = {(Tr(a_{ij}))}_{m \times n}$, are denoted as $\Bar{A}$ and $Tr(A)$, respectively. Note here that the trace of a matrix which is the sum of its diagonal elements in linear algebra is different from the trace we defined here. We also introduce a new multiplication operation called trace multiplication, denoted by $\odot_{Tr}$, which is defined by $a \odot_{Tr} b = Tr([a,b])$, where $[a,b]$ is the Euclidean inner product of the vectors $a$ and $b$ in $\mathbb{F}_4$. This operation can be extended to sets of matrices over $\mathbb{F}_4$ as follows: for any two matrices $A$ and $B$, we define $A \odot_{Tr} B = Tr(A \cdot B)$, where $\cdot$ denotes the usual matrix multiplication. The trace multiplication of two matrices exists if the usual matrix multiplication exists between them.

We now use trace multiplication to prove some results on additive codes, which help discuss other sections. If $G$ and $H$ generator and parity check matrices of an additive code $C$, then $G \odot_{Tr} H^T = 0$.

\begin{proposition}
\label{tvproduct}
Let $a=(a_0,a_1,\ldots,a_{n-1})\in \mathbb{F}_4^n$ and $b=(b_0,b_1,\ldots,b_{n-1})\in \mathbb{F}_2^n$. Then $a\odot_{Tr} b=Tr(a)\cdot b$.
\end{proposition}
\begin{proof}
 We can see that
$a\odot_{Tr} b  = Tr([a, ~b])
= Tr\left(\sum_{i=0}^{n-1} a_ib_i\right)
= \sum_{i=0}^{n-1} a_ib_i+{\left(\sum_{i=0}^{n-1} a_ib_i\right)}^2
= \sum_{i=0}^{n-1} (a_i+a_i^2)\cdot b_i
= Tr(a)\cdot b.$
\end{proof}

\begin{proposition}
\label{tmproduct}
   Let $A$ and $B$ be two matrices over $\mathbb{F}_4$ and $E$ a matrix over $\mathbb{F}_2$ such that $A\cdot E$ and $B\cdot A$ exist. Then
   \begin{enumerate}
       \item[(i)] $A \odot_{Tr} E= Tr(A) \cdot E$.
       \item[(ii)] $B\odot_{Tr} (A\cdot E)=(B\odot_{Tr} A)\cdot E$.
   \end{enumerate}
    
\end{proposition}
\begin{proof}
(i) Let $A=(a_{ij})_{m \times l}$ and $E=(e_{ij})_{l \times n}$, where $a_{ij},e_{ij}$ are entries from $\mathbb{F}_4$ and $\mathbb{F}_2$, respectively . Then $Tr(A \cdot E) = (Tr(c_{ij}))_{m \times n}$, where $c_{ij}=\sum_{k=1}^{l} a_{ik}e_{kj}$. Now,
$Tr(c_{ij}) = Tr\left(\sum_{k=1}^{l} a_{ik}e_{kj}\right)
= \sum_{k=1}^{l} a_{ik}e_{kj}+{\left(\sum_{k=1}^{l} a_{ik}e_{kj}\right)}^2
= \sum_{k=1}^{l} (a_{ik}+a_{ik}^2)\cdot e_{kj}
=\sum_{k=1}^{l} Tr(a_{ik})e_{kj}
= Tr(a_i)\cdot e_j$. Therefore $A \odot_{Tr} E= Tr(A \cdot E)=Tr(A) \cdot E$.\\[2mm]

(ii) We have $B\odot_{Tr} (A\cdot E) = Tr(B\cdot (A\cdot E)) = Tr((B\cdot A)\cdot E )=(Tr(B\cdot A))\cdot E=(B\odot_{Tr} A)\cdot E$
using (i).
\end{proof}

\begin{corollary}
\label{actdc}
    Let $C$ be an $(n,2^k)$ additive code over $\mathbb{F}_4$, whose generator and parity check matrices are $G$ and $H$, respectively. Then $C$ is an additive complementary trace dual code if and only if $\operatorname{det} (G\odot_{Tr}G^T)\neq 0$.
\end{corollary}
 \begin{proof}
 Assume that $\operatorname{det} (G\odot_{Tr}G^T)\neq 0$. Let $x\in C\cap C^{\perp_{Tr}}$. Then, $x\in C$ and $x\in C^{\perp_{Tr}}$. Therefore $x=\alpha G$ and $x=\beta H$, where $\alpha \in \mathbb{F}_2^k$, $\beta \in \mathbb{F}_2^{2n-k}$. Thus 
 $\alpha G = \beta H$,
 implies that $G^T\alpha^T = H^T\beta^T$. Consider, $(G\odot_{Tr}G^T)\cdot \alpha^T=G\odot_{Tr} (G^T\cdot \alpha^T)= G\odot_{Tr} (H^T\cdot \beta^T)=(G\odot_{Tr} H^T)\cdot \beta^T$, using Proposition \ref{tmproduct}. Since $\operatorname{det} (G\odot_{Tr}G^T)\neq 0,~\alpha=0$. This implies that $x=0$.
 
 Conversely, suppose $C$ is an additive complementary trace dual code. Assume that $\operatorname{det} (G\odot_{Tr}G^T)=0$. Then there exists a non-zero $\alpha \in \mathbb{F}_2^k$ such that $\alpha \cdot (G\odot_{Tr}G^T)=0$. Let, $vG$ be any codeword of $C$ for some $v\in \mathbb{F}_2^k$. Then $(\alpha G) \odot_{Tr} {(vG)}^T = (\alpha\cdot (G\odot_{Tr}G^T))\cdot v^T
 = 0\cdot v^T=0$.
 Therefore $\alpha G \in C^{\perp_{Tr}}$, a contradiction.
 \end{proof} 

Much like linear codes, additive codes possess important properties, such as self-orthogonality, self-duality, and the property of containing their duals with respect to the trace inner product. In the following theorem, we establish certain conditions on the generator and parity check matrices of the additive code.

 \begin{proposition}
 \label{dualc}
        Let $C$ be an additive code of length $n$ over $\mathbb{F}_4$, whose generator and parity check matrices are $G$ and $H$, respectively. Then 
        \begin{enumerate}
            \item $C^{\perp_{Tr}}\subseteq C$ if and only if $H\odot_{Tr} H^T=0$.
            \item $C\subseteq C^{\perp_{Tr}}$ if and only if $G\odot_{Tr} G^T=0$.
            \item $C= C^{\perp_{Tr}}$ if and only if $G\odot_{Tr} G^T=0$ and $H\odot_{Tr} H^T=0$.
        \end{enumerate}
        \end{proposition}
\begin{proof}

We prove that the first result and the rest follow similarly. Suppose that $C^{\perp_{Tr}}\subseteq C$. Then the row space of the generator matrix $H$ of $C^{\perp_{Tr}}$ is the $\mathbb{F}_2$-subspace of the row space of the generator matrix $G$ of $C$. We know $G\odot_{Tr} H^T=0$. Therefore $H\odot_{Tr}H^T=0$.\\

Conversely, assume that $H\odot_{Tr}H^T=0$. Let $\alpha H= x \in C^{\perp_{Tr}}$, where $\alpha \in \mathbb{F}_2^{n-k}$. Then $(H\odot_{Tr} H^T)\cdot \alpha^T=0$, which implies $H\odot_{Tr} {(\alpha H)}^T=H\odot_{Tr} x^T=0$. Therefore $ x\in C $. Then $C^{\perp_{Tr}}\subseteq C$.
\end{proof}
\begin{example}
Let $C$ be an $(3,2^4)$ additive code over $\mathbb{F}_4$ with generator matrix $G=\begin{bmatrix} \omega^2 & \omega & \omega & \omega \\ 
\omega^2 & \omega^2 & \omega & \omega \\
\omega^2 & \omega^2  & \omega^2 & \omega \end{bmatrix}^T.$
The code $C$ is dual containing, i.e., $C^{\perp_{Tr}}\subseteq C$. The parity check matrix for $C$ corresponding to the generator matrix $G$ is $H=\begin{bmatrix}1 &1&0\\ 0&1&1\end{bmatrix}$ and we have $H\odot_{Tr} H^T=0$.
\end{example}

\section{Additive conjucyclic codes}

In this section, we will be discussing additive conjucyclic codes and presenting the main results of the paper. These results will be utilized in subsequent sections to prove other results.

A code $C$ is considered a cyclic code of length $n$ over the finite field $\mathbb{F}_4$ if the cyclic shift $\sigma(c)=(c_{n-1},c_0,\ldots,c_{n-2})$ of each codeword $c=(c_0,c_1,\ldots,c_{n-1})\in C$ is also a codeword in $C$. It is well-known that a cyclic code is isomorphic to an ideal of the quotient ring $\frac{\mathbb{F}_4[x]}{\langle x^n-1 \rangle}$. Now, we define a conjucyclic code as follows: Let $C$ be an additive code of length $n$ over the finite field $\mathbb{F}_{4}$. Then $C$ is said to be an additive conjucyclic code if for each codeword $c=(c_0,c_1,\ldots,c_{n-1})\in C$, its conjucyclic shift $T(c)=(\Bar{c}_{n-1},c_0,\ldots,c_{n-2})\in C$, where $\Bar{c}_{n-1}$ is the conjugate of $c_{n-1}$ in 
$\mathbb{F}_{4}$.

Abualrub et al. \cite{ACD} provided an algebraic structure of additive conjucyclic codes using the $\mathbb{F}_2$-linear isomorphism $\Psi: \mathbb{F}_2^{2n} \mapsto \mathbb{F}_4^{n}$ such that $\Psi (u_0, u_1, \ldots, u_{n-1}, u_n, u_{n+1}, \ldots, u_{2n-1}) = (u_0+(u_0+u_n) \omega, u_1+(u_1+u_{n+1}) \omega, \ldots, u_{n-1}+(u_{n-1}+u_{2n-1})\omega)$. They proved that for every additive conjucyclic code $C$ of length $n$, there exists a cyclic code $D$ of length $2n$ such that $C=\Psi (D)$. This enables us to examine the structure of conjucyclic codes. Furthermore, they showed that there is no non-trivial linear conjucyclic code \cite[Theorem 20]{AD}. Thus, our study will solely focus on additive conjucyclic codes.

We call a pair of cyclic codes $C_1$ and $C_2$ satisfying the LCP condition, i.e., $C_1\cap C_2={0}$ and $C_1+C_2=\mathbb{F}_4^n$, as cyclic complementary pair (CCP) of codes. We first derive a condition for two cyclic codes $C_1$ and $C_2$ to be CCP using their generator polynomials.
 
 \begin{lemma}\cite{ACD}
 \label{agm}
 Let $C$ be an additive conjucyclic code of length $n$ given by $C_{g(x)}=\Psi(D_{g(x)})$, where $g(x)$ is the generating polynomial of the binary cyclic code $D$ with $deg(g(x))=k$, and $\xi_{g(x)}$ the codeword in $D$ corresponding to the generating polynomial $g(x)$ and $\eta_{g(x)}=\Psi(\xi_{g(x)})$. Then the generator matrix $G_{\eta_{g(x)}}$ of $C$ is of the form ${\begin{bmatrix}\Psi(\xi_{g(x)})& T(\Psi(\xi_{g(x)}))& \ldots &T^{2n-k-1}(\Psi(\xi_{g(x)}))\end{bmatrix}}^T$.
 \end{lemma}

 \begin{proposition}\cite[Lemma 6]{ACD}
     \label{innsep}
     Let $a$ and $b$ be two vectors in $\mathbb{F}_2^{2n}$. Then \[\langle \Psi(a),\Psi(b)\rangle=[a,b].\]
 \end{proposition}
 From \cite[Equation 7]{ACD}, we have the following equation $\Psi(\sigma^j(\xi_{g(x)}))=T^j(\eta_{g(x)})$, $\forall j \geq 1$.
 \begin{proposition}
\label{csccs}
Let $a=(a_0,a_1,\ldots,a_{n-1},a_n,a_{n+1},\ldots,a_{2n-1})\in D\subseteq \mathbb{F}_2^{2n}$ and $b=(b_0,b_1,\ldots,b_{n-1},b_n,b_{n+1},\ldots,b_{2n-1})\in D\subseteq \mathbb{F}_2^{2n}$ . Then $ T^i(\Psi(a)) \odot_{Tr} T^j(\Psi(b))=\sigma^i(a)\cdot \sigma^j(b)$ for all $i,j\geq 0$.
\end{proposition}
\begin{proof}
    From Proposition \ref{innsep} and the above equation the proof follows.
\end{proof}
\begin{proposition}\cite[Corollary $2.3$]{GKG20}
\label{lcp1}
Let $E_i$ be an $[n,k_i]$ linear code over $\mathbb{F}_q$ with generator matrix $G_i$ and parity check matrix $H_i$, respectively, for $i=1,2$. If $(E_1, E_2)$ is a linear complementary pair (LCP), then $\operatorname{rank}(G_1H_2^T)=k_1$ and $\operatorname{rank}(G_2H_1^T)=k_2$.
\end{proposition}
\begin{proposition}
\label{rpc}
Let $E_i$ be a cyclic code of length $n$ over $\mathbb{F}_2$ with generating polynomial $g_i(x)$ and parity polynomial $h_i(x)$, respectively, for $i=1,2$, such that $g_2(x)=h_1(x)$. Then $h_2^{*}(x)=g_1^{*}(x)$ and $h_1^{*}(x)=g_2^{*}(x)$ , where $g_i^{*}(x)$ and $h_i^{*}(x)$ are reciprocal polynomials of $g_i(x)$ and $h_i(x)$, respectively, for $i=1,2$.
\end{proposition}
\begin{proof}
Since $h_2(x)=(x^n+1)/g_2(x)=(x^n+1)/h_1(x)$, so $h_2^{*}(x)= (x^n+1)/h_1^{*}(x)$.
Thus $h_2^{*}(x)=g_1^{*}(x)$.
Again $g_1(x)=(x^n+1)/h_1(x)=(x^n+1)/g_2(x)$ implies that $g_1(x)=h_2(x)$.
Now $h_1(x)=(x^n+1)/g_1(x)=(x^n+1)/h_2(x)$ implies that $h_1^{*}(x)= (x^n+1)/h_2^{*}(x)$.
Thus $ h_1^{*}(x)=g_2^{*}(x)$.
\end{proof} 
The generator matrix of an $[n,k]$ cyclic code $E=\langle g(x)\rangle$, corresponding to the generating polynomial $g(x)$, is defined as $G_g={\begin{bmatrix}g & \sigma(g)& \ldots &\sigma^{k-1}(g)\end{bmatrix}}^T$, where $g$ is the vector representation of the polynomial $g(x)$. In the following theorem, we derive a condition on the generator matrices of CCP of codes.
\begin{corollary}
\label{lcp2}
Let $E_i$ be a cyclic code of length $n$ over $\mathbb{F}_2$ with generating polynomial $g_i(x)$, respectively, for $i=1,2$, such that $g_2(x)=(x^n+1)/g_1(x)$ and $\deg (g_2(x))=k$. If $(E_1,E_2)$ is a CCP, then $\operatorname{rank}(G_{g_1}G_{g_1^{*}}^T)=k$ and $\operatorname{rank}(G_{g_2}G_{g_2^{*}}^T)=n-k$.

\end{corollary}
\begin{proof}
Let $H_i$ be a parity check matrix and $h_i(x)$ be a parity check polynomial of the cyclic code $E_i$, for $i=1,2$. Then $g_2(x)=(x^n+1)/g_1(x)=h_1(x)$. Therefore from the Proposition \ref{rpc}, $h_1^{*}(x)=g_2^{*}(x)$ and $h_2^{*}(x)=g_1^{*}(x)$. We know $H_i$ is the generator matrix for the cyclic code $E_i^\perp$, corresponding to its generating polynomial $h_i^{*}(x)$, for $i=1,2$. Therefore $H_1=G_{g_2^{*}}$ and $H_2=G_{g_1^{*}}$. Then from Proposition \ref{lcp1}, $\operatorname{rank}(G_{g_1}G_{g_1^{*}}^T)=k$ and $\operatorname{rank}(G_{g_2}G_{g_2^{*}}^T)=n-k$.

\end{proof} 
The main result of this section is on the dimension of ACP of ACC codes.
 \begin{theorem}
 \label{acp}
 Let $C_i$ be an additive conjucyclic codes of length $n$ over $\mathbb{F}_4$ such that $C_i=\Psi(D_i)$, for $i=1,2$. Also let $g_1(x)$ and $g_2(x)=(x^{2n}+1)/g_1(x)$ be the generating polynomials of $D_1$ and $D_2$, respectively, with $\eta_{g_1(x)}=\Psi(\xi_{g_1(x)})$ and $\eta_{g_2(x)}=\Psi(\xi_{g_2(x)})$, where $\deg(g_2(x))=k$. If $(C_1,C_2)$ is ACP, then $\operatorname{rank}(G_{\eta_{g_1(x)}} \odot_{Tr} G_{\eta_{g_1^{*}(x)}}^T)=k$ and $\operatorname{rank}(G_{\eta_{g_2(x)}} \odot_{Tr} G_{\eta_{g_2^{*}(x)}}^T)=2n-k$. Furthermore, $C_1$ and $C_2^{\perp_{Tr}}$ are equivalent.
 \end{theorem}
 \begin{proof}
  From Lemma \ref{agm}, the generator matrix of $C_1$ is $G_{\eta_{g_1(x)}}={\begin{bmatrix}\Psi(\xi_{g_1(x)})& T(\Psi(\xi_{g_1(x)}))& \ldots &T^{k-1}(\Psi(\xi_{g_1(x)}))\end{bmatrix}}^T$ and $G_{\eta_{g_1^{*}(x)}}={\begin{bmatrix}\Psi(\xi_{g_1^{*}(x)})& T(\Psi(\xi_{g_1^{*}(x)}))& \ldots & T^{k-1}(\Psi(\xi_{g_1^{*}(x)}))\end{bmatrix}}^T$.
 Using Proposition \ref{csccs}, $G_{\eta_{g_1(x)}} \odot_{Tr} G_{\eta_{g_1^{*}(x)}}^T=G_{\xi_{g_1(x)}}\cdot G_{\xi_{g_1^{*}(x)}}^T$.
 Similarly, $G_{\eta_{g_2(x)}} \odot_{Tr} G_{\eta_{g_2^{*}(x)}}^T=G_{\xi_{g_2(x)}}\cdot G_{\xi_{g_2^{*}(x)}}^T$.
 From Corollary \ref{lcp2}, $\operatorname{rank}(G_{\eta_{g_1(x)}} \odot_{Tr} G_{\eta_{g_1^{*}(x)}}^T)=\operatorname{rank}(G_{\xi_{g_1(x)}}\cdot G_{\xi_{g_1^{*}(x)}}^T)=k$ and $\operatorname{rank}(G_{\eta_{g_2(x)}} \odot_{Tr} G_{\eta_{g_2^{*}(x)}}^T)=\operatorname{rank}(G_{\xi_{g_2(x)}}\cdot G_{\xi{g_2^{*}(x)}}^T)=2n-k$. From \cite[Theorem 2.4]{CG18}, we have $D_1$ and $D_2^\perp$ are equivalent, so are $C_1$ and $C_2^{\perp_{Tr}}$ as $C_1=\Psi(D_1)$ and $C_2^{\perp_{Tr}} = \Psi(D_2^\perp)$.
 \end{proof}
 We now illustrate the results discussed above with an example.
\begin{example}
We know that $x^{14}+1=(x+1)^2(x^3+x^2+1)^2(x^3+x+1)^2\in \mathbb{F}_2[x]$. Consider $g_1(x)=(x+1)^2(x^3+x+1)^2$ be the generating polynomial of a cyclic code $C$ of length $2n$ over $\mathbb{F}_2$. Therefore $g_2(x)=(x^3+x^2+1)^2$ and consecutively $g_1^{*}(x)=(x+1)^2(x^3+x^2+1)^2$ and $g_2^{*}(x)=(x^3+x+1)^2$. Now, $\eta_{g_1(x)}=\Psi(\xi_{g_1(x)})=(\omega^2,\omega,0,0,\omega^2,0,\omega^2)$ and $\eta_{g_2(x)}=\Psi(\xi_{g_2(x)})=(\omega^2,0,\omega^2,0,0,0,\omega^2)$. Also $\eta_{g_1^{*}(x)}=(\omega^2,\omega,\omega^2,0,\omega^2,0,0)$ and $\eta_{g_2^{*}(x)}=(\omega^2,0,0,0,\omega^2,0,\omega^2)$. Thus \\
\scalebox{0.78}{
$G_{\eta_{g_1(x)}} \odot_{Tr} G_{\eta_{g_1^{*}(x)}}^T=\begin{bmatrix} 1&0&0&0&1&0\\
0&1&0&0&0&1\\ 0&0&1&0&0&0\\ 0&0&0&1&0&0\\ 0&0&0&0&1&0\\ 0&0&0&0&0&1\end{bmatrix}$ and $G_{\eta_{g_2(x)}} \odot_{Tr} G_{\eta_{g_2^{*}(x)}}^T=\begin{bmatrix} 0&0&0&0&0&0&1&0\\ 0&0&0&0&0&0&0&1\\ 1&0&0&0&0&0&0&0\\0&1&0&0&0&0&0&0\\0&0&1&0&0&0&0&0\\0&1&0&1&0&0&0&0\\1&0&0&0&1&0&0&0\\0&1&0&0&0&1&0&0 \end{bmatrix}.$}\\[2mm]
Hence $\operatorname{rank}(G_{\eta_{g_1(x)}} \odot_{Tr} G_{\eta_{g_1^{*}(x)}}^T)=6$ and $\operatorname{rank}(G_{\eta_{g_2(x)}} \odot_{Tr} G_{\eta_{g_2^{*}(x)}}^T)=14-6=8$.
\end{example} 

 We define the Euclidean hull of a linear code $D$ over $\mathbb{F}_2$ as $\mathcal{H}=D \cap D^\perp$, i.e., intersection of $C$ with its dual $C^{\perp}$, where $D^\perp$ is the Euclidean dual of $D$. Similarly, we define hull of an ACC code $C$ over $\mathbb{F}_4$ with respect to the trace dual as the intersection of $C$ and its trace dual $C^{\perp_{Tr}}$, and is denoted by $\mathcal{H}_{Tr}$, i.e., $\mathcal{H}_{Tr}=C \cap C^{\perp_{Tr}}$. In this section, a condition on the dimension of the hull of an ACC code is derived.
\begin{proposition}   
\label{hulliso}
Let $C\subseteq \mathbb{F}_4^n$ be an ACC code over $\mathbb{F}_4$, then $\mathcal{H}_{Tr}=C \cap C^{\perp_{Tr}}=\Psi (D \cap D^\perp)=\Psi (\mathcal{H})$.
\end{proposition}
\begin{proof}
Let $\Psi(a)\in \Psi(D\cap D^\perp)$. Then $a\in D$ and $a\in D^\perp$. This further implies that
$\Psi(a)\in \Psi(D)\cap \Psi(D^\perp)$. Therefore $\Psi(a)\in C \cap C^{\perp_{Tr}}$. Hence $\Psi(D \cap D^\perp)\subseteq C \cap C^{\perp_{Tr}}$.\\
Again, let $b\in C \cap C^{\perp_{Tr}}$. Then $b\in C=\Psi(D)$ and $b\in \Psi(D^\perp)= C^{\perp_{Tr}}$.
Thus $\Psi^{-1}(b)\in D\cap D^\perp$, which
implies that $ b\in \Psi(D\cap D^\perp)$ and therefore $C \cap C^{\perp_{Tr}}\subseteq \Psi(D \cap D^\perp)$.
Hence $ C \cap C^{\perp_{Tr}}=\Psi(D \cap D^\perp)$.
\end{proof}
In the following lemma, we find a generator matrix of the hull of an ACC code over $\mathbb{F}_4$.

\begin{proposition}
\label{hullg1}
Let $C=\Psi(D)\subseteq \mathbb{F}_4^n$ be an additive conjucyclic code over $\mathbb{F}_4$ of length $n$. Also, let $g(x)$ be the generating polynomial of the cyclic code $D$ and $h^*(x)$ be the reciprocal polynomial of $h(x)=\frac{x^{2n}+1}{g(x)}$. Let $p(x)=l.c.m\ \{g(x),\ h^*(x)\}$ and $deg\ p(x)=s$. Then $p(x)$ is the generating polynomial of the hull $\mathcal{H}=D\cap D^\perp$ and $\operatorname{dim} (\mathcal{H})=2n-s$. Also if $\eta_{p(x)}=\Psi (\xi_{p(x)})$, where $\xi_{p(x)}$ be the vector corresponding to the polynomail $p(x)$ in $D\subseteq \mathbb{F}_2^{2n}$, then an additive generator matrix $G_{Tr}$ of $\mathcal{H}_{Tr}$ is given by $G_{Tr}={\begin{bmatrix} \eta_{p(x)}&T(\eta_{p(x)})& \ldots &  T^{2n-s-1}(\eta_{p(x)}) \end{bmatrix}}^T$.
\end{proposition}
\begin{proof}
 Since $\Psi (\mathcal{H})=\mathcal{H}_{Tr}$ and $p(x)$ is the generating polynomial of the hull $\mathcal{H}$, so $\eta_{p(x)}=\Psi (\xi_{p(x)})$ is a generator vector for $\mathcal{H}_{Tr}$. 
 
\end{proof} 

We now present another important result of this section. In the following theorem, a necessary and sufficient condition is obtained on the $\mathbb{F}_2$- dimension of the hull of an additive conjucyclic code with respect to trace dual.
\begin{theorem}
\label{hulld}
  Let $C$ be an $(n,2^k)$ additive conjucyclic code over $\mathbb{F}_4$ with generator matrix $G$. Then $\operatorname{dim}_{\mathbb{F}_2}(C\cap C^{\perp_{Tr}})=p$ if and only if $\operatorname{rank}(G \odot_{Tr} G^T)=k-p$.
\end{theorem}
\begin{proof}
Let us consider $C$ is an $(n,2^k)$ additive conjucyclic code over $\mathbb{F}_4$ with generator matrix $G$. Then there exist a $[2n,k]$ cyclic code over $\mathbb{F}_2$ such that $C=\Psi(D)$. Suppose that $\operatorname{dim}_{\mathbb{F}_2}(C\cap C^{\perp_{Tr}})=p$. Then there exists $\{a_1,a_2,\ldots,a_{p}\}$ a basis set for $C\cap C^{\perp_{Tr}}$ over $\mathbb{F}_2$. Then we can have a generator matrix $G$ of the ACC code $C$ over $\mathbb{F}_4$ as
$G={\begin{bmatrix} a_1 & a_2 & \ldots & a_{p} & a_{p+1}  & \ldots & a_{p+(k-p)} \end{bmatrix}}^T$. Therefore
\begin{eqnarray*}
Tr(GG^T) &=& Tr~\begin{bmatrix} a_1\cdot a_1 & a_1\cdot a_2 & \ldots & a_1\cdot a_{p+(k-p)} \\
a_2\cdot a_1 & a_2\cdot a_2 & \ldots & a_2\cdot a_{p+(k-p)}\\
\vdots & \vdots & \ddots & \vdots \\
a_{p+(k-p)}\cdot a_1 & a_{p+(k-p)}\cdot a_2 & \ldots & a_{p+(k-p)}\cdot a_{p+(k-p)}
\end{bmatrix}\\
&=& Tr~ \begin{bmatrix} O_{p\times p} & O_{p\times (k-p)}\\ O_{(k-p)\times p} & P_{(k-p)\times (k-p)}
\end{bmatrix}
~ \mbox{as}~ a_i\cdot a_j=0 ~ ~ \mbox{for all}~ 1\leq i \leq p,~ 1\leq j\leq k,
\end{eqnarray*}
where $P_{(k-p)\times (k-p)}= \begin{bmatrix} a_{p+1}\cdot a_{p+1} & a_{p+1}\cdot a_{p+2} & \ldots & a_{p+1}\cdot a_{p+(k-p)} \\
a_{p+2}\cdot a_{p+1} & a_{p+2}\cdot a_{p+2} & \ldots & a_{p+2}\cdot a_{p+(k-p)}\\
\vdots & \vdots & \ddots & \vdots \\
a_{p+(k-p)}\cdot a_{p+1} & a_{p+(k-p)}\cdot a_{p+2} & \ldots & a_{p+(k-p)}\cdot a_{p+(k-p)}
\end{bmatrix}$.\\\\
Now, 
$Tr(P_{(k-p)\times (k-p)}) \\=\begin{bmatrix} Tr(a_{p+1}\cdot a_{p+1}) & Tr(a_{p+1}\cdot a_{p+2}) & \ldots & Tr(a_{p+1}\cdot a_{p+(k-p)}) \\
Tr(a_{p+2}\cdot a_{p+1}) & Tr(a_{p+2}\cdot a_{p+2}) & \ldots & Tr(a_{p+2}\cdot a_{p+(k-p)})\\
\vdots & \vdots & \ddots & \vdots \\
Tr(a_{p+(k-p)}\cdot a_{p+1}) & Tr(a_{p+(k-p)}\cdot a_{p+2}) & \ldots & Tr(a_{p+(k-p)}\cdot a_{p+(k-p)})
\end{bmatrix}\\
= \begin{bmatrix} a_{p+1}\odot_{Tr} a_{p+1} & a_{p+1}\odot_{Tr} a_{p+2} & \ldots & a_{p+1}\odot_{Tr} a_{p+(k-p)} \\
a_{p+2}\odot_{Tr} a_{p+1} & a_{p+2}\odot_{Tr} a_{p+2} & \ldots & a_{p+2}\odot_{Tr} a_{p+(k-p)}\\
\vdots & \vdots & \ddots & \vdots \\
a_{p+(k-p)}\odot_{Tr} a_{p+1} & a_{p+(k-p)}\odot_{Tr} a_{p+2} & \ldots & a_{p+(k-p)}\odot_{Tr} a_{p+(k-p)}
\end{bmatrix}\\
= \begin{bmatrix} a_{p+1}^\prime\cdot a_{p+1}^\prime & a_{p+1}^\prime\cdot a_{p+2}^\prime & \ldots & a_{p+1}^\prime\cdot a_{p+(k-p)}^\prime \\
a_{p+2}^\prime\cdot a_{p+1}^\prime & a_{p+2}^\prime\cdot a_{p+2}^\prime & \ldots & a_{p+2}^\prime\cdot a_{p+(k-p)}^\prime\\
\vdots & \vdots & \ddots & \vdots \\
a_{p+(k-p)}^\prime\cdot a_{p+1}^\prime & a_{p+(k-p)}^\prime\cdot a_{p+2}^\prime & \ldots & a_{p+(k-p)}^\prime\cdot a_{p+(k-p)}^\prime
\end{bmatrix}$, where $a_r^\prime=\Psi^{-1}(a_r)\\
=A\cdot A^T$,  where $A={\begin{bmatrix} a_{p+1}^\prime & a_{p+2}^\prime & \ldots & a_{p+(k-p)}^\prime\end{bmatrix}}^T$.\\

Since $C=\Psi(D)$ therefore $\{a_1^\prime,a_2^\prime,\ldots,a_{p}^\prime\}$ a basis set for $D\cap D^{\perp}$ and the generator matrix of $D$ is $G^\prime={\begin{bmatrix} a_1^\prime & a_2^\prime & \ldots & a_{p}^\prime & a_{p+1}^\prime  & \ldots & a_{p+(k-p)^\prime} \end{bmatrix}}^T$. The rank of the matrix $A\cdot A^T$ is $\leq k-p$. Assume that $\operatorname{rank}(A\cdot A^T)<k-p$, then there exists a row, say $r^{th}$ row, that can be written as a linear combination of other rows in $A\cdot A^T$, i.e. $a_{p+r}^\prime\cdot a_{p+i}^\prime = x_1 a_{p+1}^\prime\cdot a_{p+i}^\prime + \ldots + x_{p+(r-1)} a_{p+(r-1)}^\prime\cdot a_{p+i}^\prime + x_{p+(r+1)} a_{p+(r+1)}^\prime\cdot a_{p+i}^\prime+\ldots+x_{p+(k-p)} a_{p+(k-p)}^\prime\cdot a_{p+i}^\prime,~ \mbox{for some }~ x_i\neq 0,~ \mbox{for each}~ 1\leq i \leq k-p$. We have $(a_{p+r}^\prime+x_1 a_{p+1}^\prime + \ldots + x_{p+(r-1)} a_{p+(r-1)}^\prime + x_{p+(r+1)} a_{p+(r+1)}^\prime+\ldots+x_{p+(k-p)} a_{p+(k-p)}^\prime)\cdot a_{p+i}^\prime=0 ~ \mbox{for all}~~ 1\leq i \leq k-p$. Therefore either $\alpha_r=0$ or $\alpha_r\in D^{\perp}$, where  $\alpha_r=a_{p+r}^\prime+x_1 a_{p+1}^\prime + \ldots + x_{p+(r-1)} a_{p+(r-1)}^\prime + x_{p+(r+1)} a_{p+(r+1)}^\prime+\ldots+x_{p+(k-p)} a_{p+(k-p)}^\prime$. None of this is true as $a_i$'s are linearly independent and if $\alpha_r=0$, then $a_{p+r}^\prime=x_1 a_{p+1}^\prime + \ldots + x_{p+(r-1)} a_{p+(r-1)}^\prime + x_{p+(r+1)} a_{p+(r+1)}^\prime+\ldots+x_{p+(k-p)} a_{p+(k-p)}^\prime$, which is not possible. In other case,  if $\alpha_r\in D^{\perp}$, then $a_{p+r}^\prime+x_1 a_{p+1}^\prime + \ldots + x_{p+(r-1)} a_{p+(r-1)}^\prime + x_{p+(r+1)} a_{p+(r+1)}^\prime+\ldots+x_{p+(k-p)} a_{p+(k-p)}^\prime \in D^{\perp}$, which is also not possible as no vector from $\{a_{p+1}^\prime,a_{p+2}^\prime,\ldots, a_{p+(k-p)}^\prime\}$ belongs to $D^{\perp}$. Hence  $\operatorname{rank}(A\cdot A^T)=k-p$.

Since the rank of the matrix $A\cdot A^T$ is $k-p$, so does the rank of $G \odot_{Tr} G^T$. Hence $\operatorname{rank}(G \odot_{Tr} G^T)=k-p$.\\
Conversely, let $\operatorname{rank}(G \odot_{Tr} G^T)=k-p$ and $\mathcal{G}$ be the generating matrix corresponding to $G$ of the code $D=\Psi^{-1}(C)$. Then from Proposition \ref{csccs}, $G \odot_{Tr} G^T=\mathcal{G}\cdot \mathcal{G}^T$. Thus $\operatorname{rank}(\mathcal{G}\cdot \mathcal{G}^T)=k-p$. We have from \cite[Proposition 3.1]{G2018} that  $\operatorname{rank}(\mathcal{G}\cdot \mathcal{G}^T)=\operatorname{dim}(D)- \operatorname{dim}(D\cap D^T)$. Hence $\operatorname{dim}_{\mathbb{F}_2}(C\cap C^{\perp_{Tr}})=\operatorname{dim}(D\cap D^T)=p$.
\end{proof}
We illustrate the above discussion with an example.
\begin{example}
We have, $x^{14}+1=(x+1)^2(x^3+x^2+1)^2(x^3+x+1)^2\in \mathbb{F}_2[x]$. Consider $g(x)=(x+1)^2(x^3+x+1)$ be the generating polynomial of a cyclic code $C$ of length $2n$ over $\mathbb{F}_2$, therefore $k=\operatorname{dim}_{\mathbb{F}_2}(C)=9$. Here, $\eta_{g(x)}=\Psi(\xi_{g(x)})=(\omega^2,\omega^2,\omega^2,0,0,\omega^2,0)$ and $p=\operatorname{dim}_{\mathbb{F}_2}(C\cap C^{\perp_{Tr}})=3$. Now,\\
\scalebox{0.9}{
$G \odot_{Tr} G^T=\begin{bmatrix}0&0&1&1&1&1&0&0&0\\ 0&0&0&1&1&1&1&0&0 \\ 1&0&0&0&1&1&1&1&0 \\ 1&1&0&0&0&1&1&1&1\\ 1&1&1&0&0&0&1&1&1\\ 1&1&1&1&0&0&0&1&1\\0&1&1&1&1&0&0&0&1\\0&0&1&1&1&1&0&0&0\\0&0&0&1&1&1&1&0&0 \end{bmatrix}$, so $\operatorname{rank}(G \odot_{Tr} G^T)=6=k-p$.
}
\end{example}

\section{Construction of Trace codes from the ACC codes}\label{Sec.5}
Let us recall the trace mapping $Tr:\mathbb{F}_4\rightarrow \mathbb{F}_2$, where $Tr(\alpha)=\alpha+\Bar{\alpha}$ and $\Bar{\alpha}=\alpha^2$ is the conjugate of the element $\alpha\in \mathbb{F}_4$. For any $a=(a_0,a_1,\ldots,a_{n-1})\in \mathbb{F}_4^n$, the trace of the vector $a$ is defined by $Tr(a)=(Tr(a_0),Tr(a_1),\ldots,Tr(a_{n-1}))$. The trace code of a code $C$ is defined as $Tr(C)=\{Tr(c)~:~c\in C\}$. We can obtain a corresponding trace code with the same length and lower cardinality for any ACC code. The motivation for developing such codes is that, in some circumstances, we can construct a new code with a greater minimum distance. Consider an ACC code $C$ of length $n$ with its trace code $Tr(C)$. This section illustrates the characterization of the trace code of an ACC code.

We define Gray weights of the elements of $\mathbb{F}_4$ as $w_G(0)=0$, $w_G(1+\omega)=1$, $w_G(\omega)=1$, and $w_G(1)=2$. Let $c=(c_0,c_1,\ldots,c_{n-1})\in \mathbb{F}_4^n$. Then, the Gray weight of the vector $c$ is $w_G(c)=\sum_{i=0}^{n-1}w_G(c_i)$. The minimum Gray weight of an ACC code $C$ is $w_G(C)=\operatorname{min}\{w_G(c):\ c\in C \}$. The Gray distance between two vectors $c$ and $d$ in $\mathbb{F}_4^n$ is $d_G(c,d)=w_G(c-d)$. The minimum Gray distance of the code $C$ is $d_G(C)=w_G(C)$. For any vector $c\in \mathbb{F}
_4^n$ of the form $c=(a_0+b_0\omega,a_1+b_1\omega,\ldots,a_{n-1}+b_{n-1}\omega)$, the Gray weight can be extended by the map $\phi:\mathbb{F}_4^n\rightarrow \mathbb{F}_2^{2n}$ such that $\phi(c)=(a_0,a_1,\ldots,a_{n-1},a_0+b_0,a_1+b_1,\ldots,a_{n-1}+b_{n-1})$. Note that $(\Psi \cdot \phi)(c) =\Psi (\phi(c) ) =\Psi(a_0,a_1,\ldots,a_{n-1},a_0+b_0,a_1+b_1,\ldots,a_{n-1}+b_{n-1})=(a_0+(a_0+a_0+b_0)\omega,a_1+(a_1+a_1+b_1)\omega,\ldots,a_{n-1}+(a_{n-1}+a_{n-1}+b_{n-1})\omega)=(a_0+b_0\omega,a_1+b_1\omega,\ldots,a_{n-1}+b_{n-1}\omega)=c$ for all $c\in \mathbb{F}_4^n$. This gives $\phi=\Psi^{-1}$.

The following theorem provides the characterization of the trace code of an ACC code.
\begin{proposition}
\label{acccyclic}
Let $C$ be an ACC code of length $n$ over $\mathbb{F}_4$, then $Tr(C)$ is a cyclic code of length $n$ over $\mathbb{F}_2$.
\end{proposition}
\begin{proof}
 Let $Tr(c)\in Tr(C)$ and $\sigma$ be the cyclic shift operation. Then
  $\sigma(Tr(c))=\sigma(Tr(c_0,c_1,\ldots,c_{n-1}))
 =\sigma(c_0+\Bar{c}_0,c_1+\Bar{c}_1,\ldots,c_{n-1}+\Bar{c}_{n-1})
 =(c_{n-1}+\Bar{c}_{n-1},c_0+\Bar{c}_0,\ldots,c_{n-2}+\Bar{c}_{n-2})
 =Tr(\Bar{c}_{n-1},c_0,\ldots,c_{n-2})=Tr(T(c))$.
 Since C is conjucyclic, $T(c)\in C$. Now $\sigma(Tr(c))=Tr(T(c))\in Tr(C)$ for each $c\in C$ and therefore $\sigma(Tr(C))=Tr(C)$. Hence $C$ is cyclic.
\end{proof}
\begin{proposition}
\label{trrrwc}
    Let $C$ be an ACC code of length $n$ over $\mathbb{F}_4$, then $tr(C)\subseteq C.$    
\end{proposition}
\begin{proof}
  Let $Tr(c)=(Tr(c_0),Tr(c_1),\ldots,Tr(c_{n-1}))\in Tr(C)$, where $c=(c_0,c_1,\ldots,c_{n-1})\in C$. Since $C$ is an conjucyclic code so $T^n(c)=(\Bar{c}_0,\Bar{c}_1,\ldots,\Bar{c}_{n-1})\in C$. Therefore $c+T^n(c)=(c_0+\Bar{c}_0,c_1+\Bar{c}_1,\ldots,c_{n-1}+\Bar{c}_{n-1})\in C$, i.e. $(Tr(c_0),Tr(c_1),\ldots,Tr(c_{n-1}))=Tr(c)\in C$. Thus $Tr(C)\subseteq C.$
\end{proof}

The Trace mapping defined above can be extended to $\mathbb{F}_4^n$ as $Tr:
\mathbb{F}_4^{n}\rightarrow \mathbb{F}_2^n$ by $Tr(a_0+(a_0+a_n)\omega,a_1+(a_1+a_{n+1})\omega,\ldots,a_{n-1}+(a_{n-1}+a_{2n-1})\omega)=(a_0+a_n,a_1+a_{n+1},\ldots,a_{n-1}+a_{2n-1})$. Combining the Trace mapping Tr and the $\mathbb{F}_2$-isomorphism $\Psi$, we have the following mapping defined as $\Phi=Tr\circ \Psi:
\mathbb{F}_2^{2n}\rightarrow \mathbb{F}_2^n$ by $\Phi(a_0,a_1,\ldots,a_{n-1},a_n,a_{n+1},\ldots,a_{2n-1})=(a_0+a_n,a_1+a_{n+1},\ldots,a_{n-1}+a_{2n-1})$. We can see that $\Phi$ is surjective.
\begin{theorem}
 If $D$ be a cyclic code of length $2n$ over $\mathbb{F}_2$, then $\Phi(D)$ is a cyclic code of length $n$ over  $\mathbb{F}_2$.
\end{theorem}
\begin{proof}
 Since $C=\Psi(D)$ is conjucyclic, we have $Tr(\Psi(D))$ is cyclic from Proposition \ref{acccyclic}.
\end{proof}
Let $D=\langle g(x) \rangle$ be a cyclic code of length $2n$, where $g(x)$ is the generating polynomial with $\deg (g(x))=k$. Then $\Phi(D)=\langle\Phi_v(x^i g(x))\mid 0\leq i \leq 2n-k\rangle$, where $\Phi_v(x^i g(x))$ is the vector representation of $\Phi(x^i g(x))$. Now we identify every vector $\Phi(a_0,a_1,\ldots,a_{n-1},a_n,a_{n+1},\ldots,a_{2n-1})\in \Phi(D)$ as a polynomial $(a_0+a_n)+(a_1+a_{n+1})x+\ldots+(a_{n-1}+a_{2n-1})x^{n-1}$. Then the generator polynomial of the cyclic code $\mathfrak{C}=\Phi(D)$ is $r(x)=\operatorname{gcd}(\Phi_p(g(x)),x^n+1)$, where $\Phi_p(g(x))$ is the polynomial representation of $\Phi_v(g(x))$. The generator polynomial of $\Phi(D^\perp)$ is $t(x)=\operatorname{gcd}(\Phi_p(h^{*}(x)),x^n+1)$, where $h^\ast(x)$ is the reciprocal polynomial of $h(x)=(x^{2n}+1)/g(x)$.
\begin{lemma}
\label{ctc}
Let $D$ be a cyclic code of length $2n$ over $\mathbb{F}_2$ and $D^\perp$ is the dual of $D$. Then $\Phi(a)\cdot \Phi(b)=0$ for all $a\in D,~ b\in D^\perp$.
\end{lemma}
\begin{proof}
 For any $a=(a_0,a_1,\ldots,a_{n-1},a_n,a_{n+1},\ldots,a_{2n-1})\in D$ and\\
$b=(b_0,b_1,\ldots,b_{n-1},b_n,b_{n+1},\ldots,b_{2n-1})\in D^\perp$. We have $a \cdot b=\sum_{i=0}^{2n-1}a_i b_i=0$. Since $D$ is cyclic, we have $\sigma^n(a)=(a_n,a_{n+1},\ldots,a_{2n-1},a_0,\ldots,a_{n-1})\in D$. Therefore $\sigma^n(a) \cdot b = \sum_{i=0}^{n-1}( b_i a_{n+i}+a_i b_{n+i})=0$. 

Now, we have $\Phi(a)=(a_0+a_n,a_1+a_{n+1},\ldots,a_{n-1}+a_{2n-1})$ and  $ \Phi(b) = (b_0+b_n,b_1+b_{n+1}, \ldots, b_{n-1}+b_{2n-1})$. Then $\Phi(a) \cdot \Phi(b)=\sum_{i=0}^{2n-1}a_ib_i+\sum_{i=0}^{n-1}(a_i b_{n+i}+ b_i a_{n+i})=a \cdot b + \sigma^n(a) \cdot b=0$.
\end{proof}
\begin{corollary}
\label{tracedc}
Let $C$ be an ACC code of length $n$ over $\mathbb{F}_4$ such that $C=\Psi(D)$. Then $Tr(C^{\perp_{Tr}})\subseteq {Tr(C)}^\perp$.
\end{corollary}
\begin{proof}
 Let $r\in Tr(C^{\perp_{Tr}})$. Then $r=\Phi(e)$ for some $e\in D^\perp$ as $\Phi$ is surjective. This implies that $e\cdot f=0$ for all $f\in D$. Therefore by Lemma \ref{ctc}, $\Phi(e)\cdot \Phi(f)=0$ for all $\Phi(f)\in \Phi(D)$. This implies that $\Phi(e)\in {\Phi(D)}^\perp={Tr(C)}^\perp$. Hence the theorem.
\end{proof}
The reverse inclusion of the Corollary \ref{tracedc} is not true in general. We provide a counter-example in the following.
\begin{example}
Let us consider $C$ be an ACC code of length $7$ such that $C=\Psi(D)$, where $D=\langle g(x)\rangle$ and $g(x)=(1+x)^2(1+x+x^3)$. Then $\Phi_p(g(x))=(1+x)^2(1+x+x^3)$ and $r(x)=\operatorname{gcd}(\Phi_p(g(x)),x^7+1)=(1+x)(1+x+x^3)$. Therefore $\Phi(D)=\langle r(x) \rangle$ and $\Phi(D)^\perp={Tr(C)}^\perp=\langle \frac{x^7+1}{r^{*}(x)}\rangle=\langle (1+x+x^3) \rangle$. We have $D^\perp=\langle h^{*}(x) \rangle$, where $h^{*}(x)=(1+x+x^3)^2(1+x^2+x^3)$. Then $\Phi_p(h^{*}(x))=(1+x+x^3)(1+x^2+x^3)$ and $t(x)=\operatorname{gcd}(\Phi_p(h^{*}(x)),x^7+1)=(1+x+x^3)(1+x^2+x^3)$. We have $\Phi(D^\perp)=Tr(C^{\perp_{Tr}})=\langle t(x) \rangle=\langle (1+x+x^3)(1+x^2+x^3) \rangle\subseteq {Tr(C)}^\perp$, but converse ${Tr(C)}^\perp \subset Tr(C^{\perp_{Tr}})$ does not hold.

\end{example}

However, the reverse inclusion of the Corollary \ref{tracedc} is true under some conditions. 
We now establish such a condition in the following theorem.
\begin{theorem}
\label{tdc}
Let $C$ be an ACC code of length $n$ over $\mathbb{F}_4$ such that $C=\Psi(D)$, where $D=\langle g(x)\rangle$, $r(x)=\operatorname{gcd}(\Phi_p(g(x)),x^n+1)$ and $t(x)=\operatorname{gcd}(\Phi_p(h^{*}(x)),x^n+1)$. If $t(x)\mid \frac{x^n+1}{r^{*}(x)}$, then $Tr(C^{\perp_{Tr}})={Tr(C)}^\perp$.
\end{theorem}
\begin{proof}
 Since $t(x)\mid \frac{x^n+1}{r^{*}(x)}$, so $\langle\frac{x^n+1}{r^{*}(x)}\rangle \subseteq \langle t(x)\rangle$. This implies ${Tr(C)}^\perp \subseteq Tr(C^{\perp_{Tr}})$. From Corollary \ref{tracedc}, $Tr(C^{\perp_{Tr}})={Tr(C)}^\perp$.
\end{proof}
\begin{example}
Let us consider $C$ be an ACC code of length $7$ such that $C=\Psi(D)$, where $D=\langle g(x)\rangle$ and $g(x)=1+x^2+x^4+x^8$. Then $\Phi_p(g(x))=(1+x)(1+x^2+x^3)$ and $r(x)=\operatorname{gcd}(\Phi_p(g(x)),x^7+1)=(1+x)(1+x^2+x^3)$. Therefore $\Phi(D)=\langle r(x) \rangle$ and $\Phi(D)^\perp={Tr(C)}^\perp=\langle \frac{x^7+1}{r^{*}(x)}\rangle=\langle (1+x^2+x^3) \rangle$. We have $D^\perp=\langle h^{*}(x) \rangle$, where $h^{*}(x)=1+x^4+x^6$. Then $\Phi_p(h^{*}(x))=1+x^4+x^6$ and $t(x)=\operatorname{gcd}(\Phi_p(h^{*}(x)),x^7+1)=1+x^2+x^3$. We have $\Phi(D^\perp)=Tr(C^{\perp_{Tr}})=\langle t(x) \rangle=\langle (1+x^2+x^3) \rangle={Tr(C)}^\perp$ and $t(x)\mid \frac{x^7+1}{r^{*}(x)}$.

\end{example}
In the following theorem, we show that the trace code of an ACD ACC code is an LCD code under a certain condition.
\begin{theorem}
Let $C=\Psi(D)$ be an additive complementary dual conjucyclic code of length $n$ over $\mathbb{F}_4$, where $D=\langle g(x)\rangle$, $r(x)=\operatorname{gcd}(\Phi_p(g(x)),x^n+1)$ and $t(x)=\operatorname{gcd}(\Phi_p(h^{*}(x)),x^n+1)$. If $t(x)\mid \frac{x^n+1}{r^{*}(x)}$, then $Tr(C)$ is a linear complementary dual code of length $n$ over $\mathbb{F}_2$.
\begin{proof}
Suppose $C=\Psi(D)$ is an additive complementary dual conjucyclic code over $\mathbb{F}_4$, i.e., $C\cap C^{\perp_{Tr}}=\{0\}$. From Theorem \ref{tdc}, we have $Tr(C^{\perp_{Tr}})={Tr(C)}^\perp$. From Proposition \ref{trrrwc}, we have  $Tr(C^{\perp_{Tr}})\subseteq C^{\perp_{Tr}}$. Therefore $Tr(C)\cap Tr(C)^\perp=Tr(C) \cap Tr(C^{\perp_{Tr}})\subseteq C \cap C^{\perp_{Tr}}=\{0\}$, i.e., $Tr(C)\cap Tr(C)^\perp=\{0\}$. Hence $Tr(C)$ is an LCD code.

\end{proof}
\end{theorem}
\begin{example}
Let us consider $C$ be an ACC code of length $5$ such that $C=\Psi(D)$, where $D=\langle g(x)\rangle$ and $g(x)=(1+x+x^2+x^3+x^4)^2$. Then $\Phi_p(g(x))=(1+x+x^2+x^3+x^4)$ and $r(x)=\operatorname{gcd}(\Phi_p(g(x)),x^5+1)=(1+x+x^2+x^3+x^4)$. Thus $Tr(C)=\langle r(x)\rangle$ and since $r(x)$ is a self-reciprocal polynomial therefore $Tr(C)$ is a cyclic linear complementary dual code of length $5$ over $\mathbb{F}_2$. We have $D^\perp=\langle h^{*}(x) \rangle$, where $h^{*}(x)=(1+x)^2$. Then $\Phi_p(h^{*}(x))=1+x^2$, $t(x)=\operatorname{gcd}(\Phi_p(h^{*}(x)),x^5+1)=1+x$ and $t(x)\mid \frac{x^5+1}{r^{*}(x)}$.

\end{example}

We present some optimal codes which are trace codes of ACC codes over $\mathbb{F}_4$ in Table 1.

\begin{table}[h]
    \centering
 \caption{Some Trace codes with better minimum distance from ACC codes}
 \scalebox{0.78}{
\begin{tabular}{|m{0.5cm}|m{20em}|m{7em}|m{8em}|m{5em}|}
\hline
Sl No. & Generator vector & ACC codes & Trace codes  & Remark\\
& (ACC codes)& $[n,\operatorname{dim}_{\mathbb{F}_2}(C),d_G]$ & $[n,\operatorname{dim}(Tr(C)),d_H]$ &  \\
\hline
$1$& $(\omega^2,0,\omega^2,\omega^2,\omega^2,0,0)$& $[7,10,2]$ & $[7,3,4]$ & Optimal code\\
\hline
$2$& $(\omega^2,\omega^2,0,0,0,\omega^2,\omega^2,0,0,0)$& $[10,14,2]$ & $[10,4,4]$ & Optimal code\\
\hline
$3$& $(\omega^2,0,0,\omega^2,\omega^2,\omega^2,0,0,0,0,0,0,0,0)$& $[14,23,2]$ & $[14,9,4]$ & Optimal code \\
\hline
$4$ & $(\omega^2,0,\omega^2,0,\omega^2,\omega^2,0,\omega^2,0,$&$[18,28,2]$& $[18,10,4]$& Optimal code \\
& $0,0,0,0,0,0,0,0,0)$ & & &\\
\hline
$5$ & $(\omega^2,\omega^2,\omega^2,\omega^2,0,\omega^2,\omega^2,0,$&$[23,33,4]$& $[23,11,8]$& Optimal code\\
& $0,0,0,\omega^2,0,0,0,0,0,0,0,0,0,0,0)$ & & & \\
\hline
$6$ & $(\omega^2,0,0,0,0,0,\omega^2,\omega^2,0,0,0,\omega^2,0,0,\omega^2,\omega^2,$&$[31,45,4]$& $[31,15,8]$& Optimal code\\
& $\omega^2,\omega^2,0,0,0,0,0,0,0,0,0,0,0,0,0)$ & & & \\
\hline
$7$ & $(\omega^2,0,\omega,\omega,\omega^2,0,\omega,\omega,\omega^2,0,\omega,0,0,0,0,0,0,$&$[34,59,2]$& $[34,24,4]$& Optimal code\\
& $0,0,0,0,0,0,0,0,0,0,0,0,0,0,0,0,0)$ & & & \\
\hline
$8$ & $(\omega^2,0,0,\omega,\omega,0,\omega^2,\omega^2,0,0,0,0,0,0,0,0,0,0,$&$[35,63,2]$& $[35,28,4]$& Optimal code\\
& $0,0,0,0,0,0,0,0,0,0,0,0,0,0,0,0,0)$ & & & \\

\hline
$9$ & $(\omega^2,\omega^2,\omega,0,\omega^2,0,0,0,\omega^2,0,0,\omega^2,0,0,$&$[43,56,4]$& $[43,14,14]$& Optimal code\\
& $\omega^2,\omega,0,0,\omega^2,0,0,\omega^2,0,0,0,\omega^2,0,\omega^2,\omega^2,\omega,0,$ & & & \\
& $0,0,0,0,0,0,0,0,0,0,0,0,)$&&&\\
\hline
$10$ & $(\omega^2,0,0,0,\omega^2,\omega^2,0,\omega^2,0,\omega^2,0,\omega^2,\omega^2,\omega^2,\omega^2,$&$[73,117,4]$& $[73,45,10]$& Optimal code\\
& $\omega^2,0,0,\omega^2,0,0,0,0,\omega^2,\omega^2,\omega^2,0,\omega^2,0,\omega^2,0,0,$ & & & \\
& $0,0,0,0,0,0,0,0,0,0,0,0,0,0,0,0,0,0,0,0,0,$&&&\\
& $0,0,0,0,0,0,0,0,0,0,0,0,0,0,0,0,0,0,0,0)$&&&\\
\hline
\end{tabular}}
\end{table}
\section{An application to Entanglement-assisted quantum error-correcting codes}
\label{Sec.8}

In this section, we will construct entanglement-assisted quantum error-correcting (EAQEC) codes. We will begin by using ACC codes and their trace codes to construct EAQEC codes. First, let us introduce some basic concepts and notations for quantum codes. The notation $[[n, k, d]]_q$ denotes a $q$-ary quantum code with length $n$, dimension $k$, and minimum distance $d$. An $[[n, k, d; c]]_q$ EAQEC code over $\mathbb{F}$ encodes $k$ logical qubits into $n$ physical qubits with the assistance of $c$ copies of maximally entangled states. If $c=n-k$, then it is a maximal entanglement EAQEC code. The performance of an EAQEC code is measured using the rate $\frac{k}{n}$ and net rate $\frac{k-c}{n}$. The net rate describes the rate of an EAQEC code when used as a catalytic quantum error-correcting code to create $c$ new bits of shared entanglement. The net rate of an EAQEC code can be positive, negative, or zero.
EAQEC codes with positive net rates are used as catalytic codes in quantum computing \cite{TD06}. An EAQEC code is a standard stabilizer code if $c=0$. EAQEC codes can be regarded as generalized quantum codes. In \cite{PhysRevA.103.L020601}, Markus Grassl proved that the Singleton bound for EAQEC codes is generally incorrect. Wilde and Brun provide an EAQEC code construction from classical linear codes in \cite{WB08}.

The following proposition presents a method to construct an EAQEC code from two linear codes over a finite field.
\begin{proposition}\cite[Corollary 1]{WB08}
\label{eaqebase}
Let $C_1: [n,k_1,d_1]_q$ and $C_2: [n,k_2,d_2]_q$ be two linear codes with parity check matrices $H_1$ and $H_2$, respectively. Then there exists an EAQEC code $[[n,k_1+k_2-n+c, \operatorname{min}\{d_1, d_2\}; c]]_q$, $c = \operatorname{rank}(H_1H_2^T)$ is the required number of maximally entangled states.
\end{proposition}
\begin{lemma}\cite[ Proposition 3.1]{G2018}
\label{eaqehull}
Let $C:[n,k,d]$ be a linear code over $\mathbb{F}_2$ with parity check matrix $H$. Then $\operatorname{rank}(HH^T)=\operatorname{dim}(C^\perp)-\operatorname{dim}(C\cap C^\perp)$.
\end{lemma}
The following theorem provides the construction of EAQEC codes from the trace codes of ACC codes. Consider $H$ is the parity matrix of the trace code of ACC code $C$.
\begin{theorem}
 \label{eaql}
 Let $Tr(C):[n,k,d]$ be a Trace code of a code $C$ over $\mathbb{F}_4$ with parity check matrix $H$. Therefore there exists $[[n,k-\operatorname{dim}(\operatorname{hull}(Tr(C))),d; c]]_2$ EAQEC code, where $c = \operatorname{rank}(HH^T)$ is the required number of maximally entangled states.
\end{theorem}
\begin{proof}
    From Proposition \ref{eaqebase} and Lemma \ref{eaqehull}, set $C_1=C_2=Tr(C)$, the result holds.
\end{proof}

Table 2 presents some maximal EAQEC codes with a good minimum distance from the trace codes of ACC codes. 

\begin{landscape}
\begin{table}[h]
    \centering
\caption{Some maximal EAQECCs codes from the Trace code of ACC codes}
 \vspace{2mm}
 \scalebox{1}{
\begin{tabular}{| m{0.4cm} | m{20em}| m{1.5cm} | m{13em} | m{10em}| m{4em} |}
\hline
Sl No. & Generator vector (ACC codes) & Trace codes $[n,k,d_H]$ & New EAQEC codes $[[n,k-\operatorname{dim}(\operatorname{hull}(Tr(C))),d_H;c]]_2$  & Existing EAQEC codes (from literature) & Remark \\
\hline 
\vspace{2mm}
$1$ & $(\omega^2,\omega,0)$& $[3,2,2]$ & $[[3,2,2;1]]$ & $[[3,2,2;1]]$ & $R_1$ \\ 
\hline
\vspace{2mm}

$2$ & $(\omega^2,\omega^2,0,\omega^2,\omega^2,0,\omega^2,\omega^2,0)$ & $[9,2,6]$ & $[[9,2,6;7]]$& $[[9,2,6;7]]$, {\cite{lu2015maximal}} & $R_1$ \\
\hline
\vspace{2mm}
$3$ & $(\omega^2,\omega,0,\omega^2,\omega,0,\omega^2,\omega,0,\omega^2,\omega,0,\omega^2,\omega,0)$ & $[15,2,10]$ & $[[15,2,10;13]]$ & $[[15,2,10;13]]$, {\cite{lu2015maximal}} & $R_1$\\
\hline
\vspace{5mm}

$4$ & $(\omega^2,0,\omega,0,0,0,\omega^2,0,\omega,0,0,0,\omega^2,0,\omega,$ $0,0,0,\omega^2,0,\omega,0,0,0,\omega^2,0,\omega)$ &$[30,4,10]$ & $[[30,4,10;26]]$ & $[[30,3,9;27]]$, {\cite{liu2019new}} & $R_2$\\[3mm]
\hline
\vspace{5mm}
$5$ & $(\omega^2,0,0,0,0,0,\omega,0,\omega^2,0,\omega^2,0,\omega^2,0,0,0,0,0,$ $\omega^2,0,0,0,0,0,0,0,0,0,0,0,0,0,0,0$ &$[34,16,6]$&$[[34,16,6;18]]$& $[[34,16,4;18]]$, \cite{liu2019new} & $R_3$\\[3mm]
\hline
\vspace{5mm}
$6$ & $(\omega^2,\omega,0,0,\omega^2,\omega,0,\omega^2,0,\omega^2,0,\omega^2,\omega^2,\omega,0,0,$ $0,\omega^2,\omega,0,\omega^2,0,\omega^2,0,\omega^2,\omega,0,\omega^2$ $0,0,0,0,0,0,0,0,0,0,0)$&$[39,12,12]$&$[[39,12,12;27]]$& $[[39,9,8;30]]$, \cite{liu2019new}& $R_2$\\[3mm]
\hline
\vspace{5mm}
$7$ & $(\omega^2,\omega,0,\omega,0,0,0,0,0,0,0,0,\omega^2,\omega,0,\omega,$ $0,0,0,0,0,0,0,0,0,0,0,0,0,0,$ $0,0,0,0,0,0,0,0,0,0,0,0,0)$ &$[43,28,6]$&$[[43,28,6;15]]$& $[[43,27,4;15]]$, \cite{liu2019new} &  $R_4$\\[3mm]
\hline
\vspace{5mm}
$8$ & $(\omega^2,0,\omega,\omega^2,0,0,0,0,\omega^2,\omega^2,\omega^2,0,0,0,$ $\omega,\omega,0,0,0,\omega^2,\omega^2,\omega^2,0,0,0,0,\omega,\omega,0,\omega^2$ $0,0,0,0,0,0,0,0,0,0,0,0,0)$  &$[43,15,13]$&$[[43,15,13;28]]$& $[[43,15,9;28]]$, \cite{liu2019new} & $R_3$\\[3mm]
\hline
\end{tabular}}~\\
$R_1$ -- Best known EAQEC code so far. Our code parameters match with the known best EAQEC code.\\
$R_2$ -- Relative distance of the EAQEC code we got is better than that of the existing EAQEC Code when the entanglement is almost the same.\\
$R_3$ -- The minimum distance of the EAQEC Code we obtained is  better than that of the existing EAQEC Code.\\
$R_4$ -- The minimum distance and dimension of the EAQEC Code we got are better than that of the existing EAQEC Code.\\
\end{table}
\end{landscape}

\section{Conclusion}
 In this paper, we studied additive conjucyclic codes over $\mathbb{F}_4$ with respect to the trace dual. We obtained the trace hull and dimension of an additive conjucyclic code. Additionally, we derived a necessary and sufficient condition for a conjucyclic code to have an additive complementary dual with respect to the trace dual. Using the trace dual, we also found a condition for an additive complementary pair of conjucyclic codes over $\mathbb{F}_4$. We provided a method for constructing trace codes from additive codes and observed that the code obtained from the trace code of an ACC code has good parameters, as do the EAQEC codes. Finally, we discussed the application of EAQEC codes in our study. Using the Magma computer algebra system, we computed some EAQEC codes from trace codes of ACC codes. It would be interesting to explore the algebraic structure of conjucyclic codes using ring theory, similar to cyclic codes. Furthermore, investigating other applications of additive conjucyclic codes would also be worthwhile.

\section*{Acknowledgement}
The first author would like to thank IIIT Naya Raipur for the financial support to carry out this work. The second author is supported by the National Board of Higher Mathematics, Department of Atomic Energy, India through project No. 02011/20/2021/ NBHM(R.P)/R\&D II/8775.
\section*{Statements \& Declarations}
\subsection*{Competing Interests}
The authors have no conflict of interest to declare.
\subsection*{Availability of data and materials}
The datasets supporting the conclusions of this article are included within the article.
\bibliography{acc_code}


\begin{thebibliography}{31}
\ifx \bisbn   \undefined \def \bisbn  #1{ISBN #1}\fi
\ifx \binits  \undefined \def \binits#1{#1}\fi
\ifx \bauthor  \undefined \def \bauthor#1{#1}\fi
\ifx \batitle  \undefined \def \batitle#1{#1}\fi
\ifx \bjtitle  \undefined \def \bjtitle#1{#1}\fi
\ifx \bvolume  \undefined \def \bvolume#1{\textbf{#1}}\fi
\ifx \byear  \undefined \def \byear#1{#1}\fi
\ifx \bissue  \undefined \def \bissue#1{#1}\fi
\ifx \bfpage  \undefined \def \bfpage#1{#1}\fi
\ifx \blpage  \undefined \def \blpage #1{#1}\fi
\ifx \burl  \undefined \def \burl#1{\textsf{#1}}\fi
\ifx \doiurl  \undefined \def \doiurl#1{\url{https://doi.org/#1}}\fi
\ifx \betal  \undefined \def \betal{\textit{et al.}}\fi
\ifx \binstitute  \undefined \def \binstitute#1{#1}\fi
\ifx \binstitutionaled  \undefined \def \binstitutionaled#1{#1}\fi
\ifx \bctitle  \undefined \def \bctitle#1{#1}\fi
\ifx \beditor  \undefined \def \beditor#1{#1}\fi
\ifx \bpublisher  \undefined \def \bpublisher#1{#1}\fi
\ifx \bbtitle  \undefined \def \bbtitle#1{#1}\fi
\ifx \bedition  \undefined \def \bedition#1{#1}\fi
\ifx \bseriesno  \undefined \def \bseriesno#1{#1}\fi
\ifx \blocation  \undefined \def \blocation#1{#1}\fi
\ifx \bsertitle  \undefined \def \bsertitle#1{#1}\fi
\ifx \bsnm \undefined \def \bsnm#1{#1}\fi
\ifx \bsuffix \undefined \def \bsuffix#1{#1}\fi
\ifx \bparticle \undefined \def \bparticle#1{#1}\fi
\ifx \barticle \undefined \def \barticle#1{#1}\fi
\bibcommenthead
\ifx \bconfdate \undefined \def \bconfdate #1{#1}\fi
\ifx \botherref \undefined \def \botherref #1{#1}\fi
\ifx \url \undefined \def \url#1{\textsf{#1}}\fi
\ifx \bchapter \undefined \def \bchapter#1{#1}\fi
\ifx \bbook \undefined \def \bbook#1{#1}\fi
\ifx \bcomment \undefined \def \bcomment#1{#1}\fi
\ifx \oauthor \undefined \def \oauthor#1{#1}\fi
\ifx \citeauthoryear \undefined \def \citeauthoryear#1{#1}\fi
\ifx \endbibitem  \undefined \def \endbibitem {}\fi
\ifx \bconflocation  \undefined \def \bconflocation#1{#1}\fi
\ifx \arxivurl  \undefined \def \arxivurl#1{\textsf{#1}}\fi
\csname PreBibitemsHook\endcsname

\bibitem{ST1996}
\begin{barticle}
\bauthor{\bsnm{Steane}, \binits{A.}}:
\batitle{Multiple-particle interference and quantum error correction}.
\bjtitle{Proceedings of the Royal Society of London. Series A: Mathematical,
  Physical and Engineering Sciences}
\bvolume{452}(\bissue{1954}),
\bfpage{2551}--\blpage{2577}
(\byear{1996})
\end{barticle}
\endbibitem

\bibitem{C1997}
\begin{barticle}
\bauthor{\bsnm{Calderbank}, \binits{A.R.}},
\bauthor{\bsnm{Rains}, \binits{E.M.}},
\bauthor{\bsnm{Shor}, \binits{P.W.}},
\bauthor{\bsnm{Sloane}, \binits{N.J.}}:
\batitle{Quantum error correction and orthogonal geometry}.
\bjtitle{Physical Review Letters}
\bvolume{78}(\bissue{3}),
\bfpage{405}
(\byear{1997})
\end{barticle}
\endbibitem

\bibitem{KR2005}
\begin{barticle}
\bauthor{\bsnm{Kribs}, \binits{D.}},
\bauthor{\bsnm{Laflamme}, \binits{R.}},
\bauthor{\bsnm{Poulin}, \binits{D.}}:
\batitle{Unified and generalized approach to quantum error correction}.
\bjtitle{Physical review letters}
\bvolume{94}(\bissue{18}),
\bfpage{180501}
(\byear{2005})
\end{barticle}
\endbibitem

\bibitem{BI2008}
\begin{barticle}
\bauthor{\bsnm{Bierbrauer}, \binits{J.}},
\bauthor{\bsnm{Faina}, \binits{G.}},
\bauthor{\bsnm{Giulietti}, \binits{M.}},
\bauthor{\bsnm{Marcugini}, \binits{S.}},
\bauthor{\bsnm{Pambianco}, \binits{F.}}:
\batitle{The geometry of quantum codes}.
\bjtitle{Innovations in Incidence Geometry: Algebraic, Topological and
  Combinatorial}
\bvolume{6}(\bissue{1}),
\bfpage{53}--\blpage{71}
(\byear{2008})
\end{barticle}
\endbibitem

\bibitem{BI2014}
\begin{barticle}
\bauthor{\bsnm{Bierbrauer}, \binits{J.}},
\bauthor{\bsnm{Bartoli}, \binits{D.}},
\bauthor{\bsnm{Faina}, \binits{G.}},
\bauthor{\bsnm{Marcugini}, \binits{S.}},
\bauthor{\bsnm{Pambianco}, \binits{F.}},
\bauthor{\bsnm{Edel}, \binits{Y.}}:
\batitle{The structure of quaternary quantum caps}.
\bjtitle{Designs, codes and cryptography}
\bvolume{72}(\bissue{3}),
\bfpage{733}--\blpage{747}
(\byear{2014})
\end{barticle}
\endbibitem

\bibitem{Cal96}
\begin{barticle}
\bauthor{\bsnm{Calderbank}, \binits{A.R.}},
\bauthor{\bsnm{Shor}, \binits{P.W.}}:
\batitle{Good quantum error-correcting codes exist}.
\bjtitle{Physical Review A}
\bvolume{54}(\bissue{2}),
\bfpage{1098}
(\byear{1996})
\end{barticle}
\endbibitem

\bibitem{TD06}
\begin{barticle}
\bauthor{\bsnm{Brun}, \binits{T.}},
\bauthor{\bsnm{Devetak}, \binits{I.}},
\bauthor{\bsnm{Hsieh}, \binits{M.H.}}:
\batitle{Correcting quantum errors with entanglement}.
\bjtitle{Science (New York, N.Y.)}
\bvolume{314}(\bissue{5798}),
\bfpage{436}--\blpage{439}
(\byear{2006})
\end{barticle}
\endbibitem

\bibitem{Fan16}
\begin{barticle}
\bauthor{\bsnm{Fan}, \binits{J.}},
\bauthor{\bsnm{Chen}, \binits{H.}},
\bauthor{\bsnm{Xu}, \binits{J.}}:
\batitle{Constructions of $q$-ary entanglement-assisted quantum {MDS} codes
  with minimum distance greater than $q+1$}.
\bjtitle{Quantum Info. Comput.}
\bvolume{16}(\bissue{5–6}),
\bfpage{423}--\blpage{434}
(\byear{2016})
\end{barticle}
\endbibitem

\bibitem{Qian2015}
\begin{barticle}
\bauthor{\bsnm{Qian}, \binits{J.}},
\bauthor{\bsnm{Zhang}, \binits{L.}}:
\batitle{Entanglement-assisted quantum codes from arbitrary binary linear
  codes}.
\bjtitle{Designs, Codes and Cryptography}
\bvolume{77}(\bissue{1}),
\bfpage{193}--\blpage{202}
(\byear{2015})
\end{barticle}
\endbibitem

\bibitem{ch17}
\begin{barticle}
\bauthor{\bsnm{Chen}, \binits{J.}},
\bauthor{\bsnm{Huang}, \binits{Y.}},
\bauthor{\bsnm{Feng}, \binits{C.}},
\bauthor{\bsnm{Chen}, \binits{R.}}:
\batitle{Entanglement-assisted quantum {MDS} codes constructed from negacyclic
  codes}.
\bjtitle{Quantum Information Processing}
\bvolume{16}(\bissue{12}),
\bfpage{1}--\blpage{22}
(\byear{2017})
\end{barticle}
\endbibitem

\bibitem{ch18}
\begin{barticle}
\bauthor{\bsnm{Chen}, \binits{X.}},
\bauthor{\bsnm{Zhu}, \binits{S.}},
\bauthor{\bsnm{Kai}, \binits{X.}}:
\batitle{Entanglement-assisted quantum {MDS} codes constructed from
  constacyclic codes}.
\bjtitle{Quantum Information Processing}
\bvolume{17}(\bissue{10}),
\bfpage{273}
(\byear{2018})
\end{barticle}
\endbibitem

\bibitem{PhysRevA.103.L020601}
\begin{barticle}
\bauthor{\bsnm{Grassl}, \binits{M.}}:
\batitle{Entanglement-assisted quantum communication beating the quantum
  singleton bound}.
\bjtitle{Phys. Rev. A}
\bvolume{103},
\bfpage{020601}
(\byear{2021})
\end{barticle}
\endbibitem

\bibitem{G2018}
\begin{barticle}
\bauthor{\bsnm{Guenda}, \binits{K.}},
\bauthor{\bsnm{Jitman}, \binits{S.}},
\bauthor{\bsnm{Gulliver}, \binits{T.A.}}:
\batitle{Constructions of good entanglement-assisted quantum error correcting
  codes}.
\bjtitle{Designs, Codes and Cryptography}
\bvolume{86}(\bissue{1}),
\bfpage{121}--\blpage{136}
(\byear{2018})
\end{barticle}
\endbibitem

\bibitem{luo2019}
\begin{barticle}
\bauthor{\bsnm{Luo}, \binits{G.}},
\bauthor{\bsnm{Cao}, \binits{X.}},
\bauthor{\bsnm{Chen}, \binits{X.}}:
\batitle{{MDS} codes with hulls of arbitrary dimensions and their quantum error
  correction}.
\bjtitle{IEEE Transactions on Information Theory}
\bvolume{65}(\bissue{5}),
\bfpage{2944}--\blpage{2952}
(\byear{2018})
\end{barticle}
\endbibitem

\bibitem{WB08}
\begin{barticle}
\bauthor{\bsnm{Wilde}, \binits{M.M.}},
\bauthor{\bsnm{Brun}, \binits{T.A.}}:
\batitle{Optimal entanglement formulas for entanglement-assisted quantum
  coding}.
\bjtitle{Phys. Rev. A}
\bvolume{77},
\bfpage{064302}
(\byear{2008})
\end{barticle}
\endbibitem

\bibitem{hossain2023linear}
\begin{botherref}
\oauthor{\bsnm{Hossain}, \binits{M.A.}},
\oauthor{\bsnm{Bandi}, \binits{R.}}:
Linear $\ell$-intersection pairs of cyclic and quasi-cyclic codes over a finite
  field $\mathbb{F}_q$.
Journal of Applied Mathematics and Computing,
1--17
(2023)
\end{botherref}
\endbibitem

\bibitem{DL73}
\begin{barticle}
\bauthor{\bsnm{Delsarte}, \binits{P.}}:
\batitle{An algebraic approach to the association schemes of coding theory}.
\bjtitle{Philips Res. Rep. Suppl.}
\bvolume{10},
\bfpage{97}
(\byear{1973})
\end{barticle}
\endbibitem

\bibitem{H07}
\begin{barticle}
\bauthor{\bsnm{Huffman}, \binits{W.C.}}:
\batitle{Additive cyclic codes over $\mathbb{F}_4$}.
\bjtitle{Advances in Mathematics of Communications}
\bvolume{1}(\bissue{4}),
\bfpage{427}
(\byear{2007})
\end{barticle}
\endbibitem

\bibitem{CRS}
\begin{barticle}
\bauthor{\bsnm{Calderbank}, \binits{A.R.}},
\bauthor{\bsnm{Rains}, \binits{E.M.}},
\bauthor{\bsnm{Shor}, \binits{P.}},
\bauthor{\bsnm{Sloane}, \binits{N.J.}}:
\batitle{Quantum error correction via codes over {GF(4)}}.
\bjtitle{IEEE Transactions on Information Theory}
\bvolume{44}(\bissue{4}),
\bfpage{1369}--\blpage{1387}
(\byear{1998})
\end{barticle}
\endbibitem

\bibitem{KL17}
\begin{barticle}
\bauthor{\bsnm{Kim}, \binits{J.L.}},
\bauthor{\bsnm{Lee}, \binits{N.}}:
\batitle{Secret sharing schemes based on additive codes over {GF(4)}}.
\bjtitle{Applicable Algebra in Engineering, Communication and Computing}
\bvolume{28}(\bissue{1}),
\bfpage{79}--\blpage{97}
(\byear{2017})
\end{barticle}
\endbibitem

\bibitem{BN20}
\begin{barticle}
\bauthor{\bsnm{Benbelkacem}, \binits{N.}},
\bauthor{\bsnm{Borges}, \binits{J.}},
\bauthor{\bsnm{Dougherty}, \binits{S.T.}},
\bauthor{\bsnm{Fern{\'a}ndez-C{\'o}rdoba}, \binits{C.}}:
\batitle{On $\mathbb{Z}_2\mathbb{Z}_4$-additive complementary dual codes and
  related {LCD} codes}.
\bjtitle{Finite Fields and Their Applications}
\bvolume{62},
\bfpage{101622}
(\byear{2020})
\end{barticle}
\endbibitem

\bibitem{HC13}
\begin{barticle}
\bauthor{\bsnm{Huffman}, \binits{W.C.}}:
\batitle{On the theory of $\mathbb{F}_q$-linear $\mathbb{F}_{q^t}$-codes}.
\bjtitle{Advances in Mathematics of Communications}
\bvolume{7}(\bissue{3}),
\bfpage{349}
(\byear{2013})
\end{barticle}
\endbibitem

\bibitem{EM11}
\begin{barticle}
\bauthor{\bsnm{Ezerman}, \binits{M.F.}},
\bauthor{\bsnm{Ling}, \binits{S.}},
\bauthor{\bsnm{Sole}, \binits{P.}}:
\batitle{Additive asymmetric quantum codes}.
\bjtitle{IEEE transactions on information theory}
\bvolume{57}(\bissue{8}),
\bfpage{5536}--\blpage{5550}
(\byear{2011})
\end{barticle}
\endbibitem

\bibitem{KP03}
\begin{barticle}
\bauthor{\bsnm{Kim}, \binits{J.L.}},
\bauthor{\bsnm{Pless}, \binits{V.}}:
\batitle{Designs in additive codes over {GF(4)}}.
\bjtitle{Designs, Codes and Cryptography}
\bvolume{30}(\bissue{2}),
\bfpage{187}--\blpage{199}
(\byear{2003})
\end{barticle}
\endbibitem

\bibitem{ACD}
\begin{barticle}
\bauthor{\bsnm{Abualrub}, \binits{T.}},
\bauthor{\bsnm{Cao}, \binits{Y.}},
\bauthor{\bsnm{Dougherty}, \binits{S.T.}}:
\batitle{Algebraic structure of additive conjucyclic codes over
  $\mathbb{F}_4$}.
\bjtitle{Finite Fields and Their Applications}
\bvolume{65},
\bfpage{101678}
(\byear{2020})
\end{barticle}
\endbibitem

\bibitem{LL}
\begin{botherref}
\oauthor{\bsnm{Lv}, \binits{J.}},
\oauthor{\bsnm{Li}, \binits{R.}}:
Algebraic structure of $\mathbb{F}_q$-linear conjucyclic codes over finite
  field $\mathbb{F}_{q^2}$.
arXiv preprint arXiv:2007.03963
(2020)
\end{botherref}
\endbibitem

\bibitem{AD}
\begin{barticle}
\bauthor{\bsnm{Abualrub}, \binits{T.}},
\bauthor{\bsnm{Dougherty}, \binits{S.T.}}:
\batitle{Additive and linear conjucyclic codes over $\mathbb{F}_4$}.
\bjtitle{Advances in Mathematics of Communications}
\bvolume{16}(\bissue{1}),
\bfpage{1}
(\byear{2022})
\end{barticle}
\endbibitem

\bibitem{GKG20}
\begin{barticle}
\bauthor{\bsnm{Guenda}, \binits{K.}},
\bauthor{\bsnm{Gulliver}, \binits{T.A.}},
\bauthor{\bsnm{Jitman}, \binits{S.}},
\bauthor{\bsnm{Thipworawimon}, \binits{S.}}:
\batitle{Linear $l$-intersection pairs of codes and their applications}.
\bjtitle{Designs, Codes and Cryptography}
\bvolume{88}(\bissue{1}),
\bfpage{133}--\blpage{152}
(\byear{2020})
\end{barticle}
\endbibitem

\bibitem{CG18}
\begin{barticle}
\bauthor{\bsnm{Carlet}, \binits{C.}},
\bauthor{\bsnm{G{\"u}neri}, \binits{C.}},
\bauthor{\bsnm{{\"O}zbudak}, \binits{F.}},
\bauthor{\bsnm{{\"O}zkaya}, \binits{B.}},
\bauthor{\bsnm{Sol{\'e}}, \binits{P.}}:
\batitle{On linear complementary pairs of codes}.
\bjtitle{IEEE Transactions on Information Theory}
\bvolume{64}(\bissue{10}),
\bfpage{6583}--\blpage{6589}
(\byear{2018})
\end{barticle}
\endbibitem

\bibitem{lu2015maximal}
\begin{barticle}
\bauthor{\bsnm{Lu}, \binits{L.}},
\bauthor{\bsnm{Li}, \binits{R.}},
\bauthor{\bsnm{Guo}, \binits{L.}},
\bauthor{\bsnm{Fu}, \binits{Q.}}:
\batitle{Maximal entanglement entanglement-assisted quantum codes constructed
  from linear codes}.
\bjtitle{Quantum Information Processing}
\bvolume{14}(\bissue{1}),
\bfpage{165}--\blpage{182}
(\byear{2015})
\end{barticle}
\endbibitem

\bibitem{liu2019new}
\begin{barticle}
\bauthor{\bsnm{Liu}, \binits{X.}},
\bauthor{\bsnm{Yu}, \binits{L.}},
\bauthor{\bsnm{Hu}, \binits{P.}}:
\batitle{New entanglement-assisted quantum codes from k-galois dual codes}.
\bjtitle{Finite Fields and Their Applications}
\bvolume{55},
\bfpage{21}--\blpage{32}
(\byear{2019})
\end{barticle}
\endbibitem

\end{thebibliography}


\end{document}